\documentclass[11pt]{article}
\usepackage{bbm}

\usepackage{epsfig,endnotes}
\usepackage{multirow}
\usepackage{subfig}
\usepackage{graphicx}
\usepackage{grffile}

\newcounter{subcopyrightbox@save}
\usepackage[font=bf]{caption}
\usepackage{color, url}
\usepackage{xspace} 

\usepackage{amsmath}

\usepackage{mathrsfs}
\usepackage{amssymb}
\usepackage{amsmath}
\usepackage{amsthm}
\usepackage{epstopdf}
\usepackage{balance}
\usepackage{bm}
\newtheorem{theorem}{Theorem}
\usepackage{thm-restate,thmtools}

\newtheorem{proposition}{Proposition}
\newtheorem{definition}{Definition}

\usepackage[ruled,linesnumbered,noend]{algorithm2e}
\SetKwRepeat{Do}{do}{while}
\usepackage{mathtools}

\usepackage{geometry}
\usepackage{changepage}

\usepackage{booktabs}
\usepackage{hyperref} 
\usepackage{cleveref}

\newcommand{\myparatight}[1]{\smallskip\noindent{\bf {#1}:}~}

\newcommand{\RN}[1]{%
  \textup{\uppercase\expandafter{\romannumeral#1}}%
}

\AtBeginDocument{%
  \providecommand\BibTeX{{%
    \normalfont B\kern-0.5em{\scshape i\kern-0.25em b}\kern-0.8em\TeX}}}

\begin{document}

\begin{center}
{\LARGE{\bf{On the Intrinsic Differential Privacy of Bagging}}}

  \vspace{1cm}
\begin{tabular}{ccccc}
Hongbin Liu & & Jinyuan Jia && Neil Zhenqiang Gong\\
Duke University && Duke University && Duke University\\
hongbin.liu@duke.edu && jinyuan.jia@duke.edu && neil.gong@duke.edu\\
\end{tabular}
\end{center}
  \vspace{.5cm}


\begin{abstract}
Differentially private machine learning trains models while protecting privacy of the sensitive training data. 
The key to obtain differentially private models is to introduce noise/randomness to the training process.  
In particular, existing differentially private machine learning methods add noise to the training data, the gradients, the loss function, and/or the model itself. Bagging, a popular ensemble learning framework, randomly creates some subsamples of the training data, trains a base model for each subsample using a base learner, and takes majority vote among the base models when making predictions.  Bagging has intrinsic randomness in the training process as it randomly creates subsamples. Our major theoretical results show that such intrinsic randomness already makes Bagging differentially private without the needs of additional noise. In particular, we prove that, for any base learner, Bagging with and without replacement respectively achieves  $\left(N\cdot k \cdot \ln{\frac{n+1}{n}},1- (\frac{n-1}{n})^{N\cdot k}\right)$-differential privacy  and $\left(\ln{\frac{n+1}{n+1-N\cdot k}}, \frac{N\cdot k}{n} \right)$-differential privacy, where $n$ is the training data size, $k$ is the subsample size, and $N$ is the number of base models. Moreover, we  prove that if no assumptions about the base learner are made, our derived privacy guarantees are tight. We empirically evaluate Bagging on MNIST and CIFAR10. Our experimental results demonstrate that Bagging achieves significantly higher accuracies than  state-of-the-art differentially private machine learning methods with the same privacy budgets.  
\end{abstract}


\section{Introduction}

Machine learning has transformed various areas such as  computer vision, natural language processing, healthcare, and cybersecurity. However, since a model is essentially some aggregate of the training data, the model may reveal rich information about the training data. For instance, with access to a model or only its prediction API, \emph{model inversion}~\cite{fredrikson2015model} can reconstruct (representative) training data of the model, while \emph{membership inference}~\cite{shokri2017membership} can predict whether a given data point is among the model's training data or not. As a result, the model may compromise the privacy or confidentiality of the sensitive or proprietary training data such as electronic health records, location traces, and online digital behaviors. Moreover, various countries have passed laws such as General Data Protection Regulation (GDPR)~\cite{voigt2017eu} to  regulate and protect data privacy. Therefore, privacy-preserving machine learning that trains models while protecting privacy of the training data is gaining increasing attention in both academia and industry.  

$(\epsilon,\delta)$-differential privacy~\cite{dwork2014algorithmic} has become a de facto standard for privacy-preserving data analytics. Many studies~\cite{hamm2016learning,abadi2016deep,papernot2016semi,papernot2018scalable,jordon2018pate,xie2018differentially} have extended $(\epsilon,\delta)$-differential privacy to machine learning. Roughly speaking, a machine learning method satisfies differential privacy if the learnt model does not change much when adding or removing one example in the training data. The key idea of differentially private machine learning is to introduce noise/randomness in the training process. Specifically, existing methods introduce randomness to the training data, the gradients when stochastic gradient descent is used to learn a model, the loss function, and/or the model itself. For instance, Differentially Private Stochastic Gradient Descent (DPSGD)~\cite{abadi2016deep} introduces well-calibrated Gaussian noise to the gradient computed from a random batch of the training data in each iteration when using stochastic gradient descent to learn a model. Private Aggregation of Teacher Ensembles (PATE)~\cite{papernot2016semi,papernot2018scalable} trains multiple teacher models on pre-defined disjoint chunks of the sensitive training data. Then, PATE uses the teacher models to predict labels for  some examples in a non-sensitive public dataset, aggregates the labels for each example, and adds noise to the aggregated labels to achieve differential privacy. Finally, PATE trains a student model using the non-sensitive public dataset with the aggregated labels predicted by the teacher models. The student model satisfies $(\epsilon,\delta)$-differential privacy. 

\myparatight{Our work}
Bagging~\cite{breiman1996bagging}, a popular ensemble learning framework, randomly creates some subsamples of the training data, trains a base model for each subsample using a base learner, and takes majority vote among the base models when making predictions. Bagging has intrinsic randomness in the training process as it randomly creates subsamples. Our major theoretical results have two folds. On one hand, we show that the intrinsic randomness of Bagging already makes it differentially private without the needs of additional noise. In particular, we prove that, for any base learner, {Bagging} with and without replacement respectively achieves $\left(N\cdot k \cdot \ln{\frac{n+1}{n}},1- (\frac{n-1}{n})^{N\cdot k}\right)$-differential privacy  and $\left(N \cdot\ln{\frac{n+1}{n+1- N \cdot k}},\frac{N \cdot k}{n}\right)$-differential privacy, where $n$ is the training data size, $k$ is the subsample size, and $N$ is the number of base models. Moreover, we prove that if no assumptions about the base learner are made, our derived privacy guarantees are tight. On the other hand, our theoretical results indicate that Bagging can only provide differential privacy with $\delta \geq 1/n$. According to Dwork and Roth~\cite{dwork2014algorithmic},  $\delta \geq 1/n$ provides ``just a few'' privacy guarantee,  which is equivalent to protecting the privacy of most training examples while compromising the privacy of just a few training examples. 
We empirically evaluate {Bagging} on MNIST and CIFAR10. For instance, Bagging achieves 79.55\% testing accuracy on CIFAR10 with privacy budget $\epsilon=0.2,\delta=0.18$ (i.e., $k=10000$ and $N=1$). With the same privacy budget, DPSGD~\cite{abadi2016deep} only achieves 30.63\% testing accuracy. 

Our main contributions can be summarized as follows:
\begin{itemize}
    \item We derive the intrinsic $(\epsilon,\delta)$-differential privacy of Bagging. 
    \item We prove our derived $(\epsilon,\delta)$-differential privacy of Bagging is tight if no extra assumptions about the base learner are given.
    \item We empirically compare Bagging with  state-of-the-art privacy-preserving machine learning methods on MNIST and CIFAR10.
\end{itemize}


\section{Background and Related Work}
\label{sec:background}
In this section, we first review the concept of $(\epsilon,\delta)$-differential privacy as well as its composition theorem and post-processing property. Then, we  review Bagging~\cite{breiman1996bagging} and  existing differentially private machine learning approaches.

\subsection{Differential Privacy}
Differential privacy~\cite{dwork2014algorithmic} is defined in terms of \emph{adjacent datasets}. In machine learning, a dataset consists of  training examples. We call two training datasets {adjacent datasets} if there only exists one training example that appears in one dataset but is absent in the other. With the definition of {adjacent datasets}, we can introduce the definition of $(\epsilon,\delta)$-differential privacy as follows:

\begin{definition}[$(\epsilon,\delta)$-differential privacy~\cite{dwork2014algorithmic}]\label{def-dp}
A randomized mechanism $\mathcal{M:D \rightarrow R}$ satisfies $(\epsilon,\delta)$-differential privacy if for all adjacent datasets ${D}$, ${D}^{\prime} \in \mathcal{D}$, and for all $S \subseteq \mathcal{R}$, it holds that:  
\begin{align}
    \mathrm{Pr}(\mathcal{M}(D)\in S) \leq e^{\epsilon}\cdot \mathrm{Pr}(\mathcal{M}(D^{\prime})\in S) + \delta, 
\end{align}
where the randomness is taking over the mechanism $\mathcal{M}$. 
\end{definition}

In machine learning, the randomized mechanism $\mathcal{M}$ denotes the algorithm to train a model. $(\epsilon,\delta)$-differential privacy formalizes that the learnt model does not change much when adding or removing an arbitrary example in the training data. $\epsilon$ and $\delta$ quantify the upper bound of observable probability differences between the learnt models conditioned on adjacent datasets. If $\delta=0$, we say that $\mathcal{M}$ is $\epsilon$-differentially private~\cite{dwork2014algorithmic}. Thereby, the additive term $\delta$ is considered as the probability at which the $\epsilon$-differential privacy guarantee may be broken. 

When several differential privacy mechanisms are composed, the differential privacy guarantee of the composed mechanism is the sum of the privacy guarantees of the individual mechanisms. Formally, differential privacy follows the following standard composition theorem~\cite{dwork2014algorithmic}:
\begin{theorem}
[Composition theorem of $(\epsilon,\delta)$-differential privacy~\cite{dwork2014algorithmic}]
\label{composition theorem}
Let $\mathcal{M}_{i}: \mathcal{D} \rightarrow \mathcal{R}_{i}$ be an $\left(\epsilon_{i}, \delta_{i}\right)$-differentially private
algorithm for $i \in[k] .$ If $\mathcal{M}_{[k]}: \mathcal{D} \rightarrow \prod_{i=1}^{k} \mathcal{R}_{i}$ is defined to
be $\mathcal{M}_{[k]}(\mathcal{D})=\left(\mathcal{M}_{1}(\mathcal{D}), \cdots, \mathcal{M}_{k}(\mathcal{D})\right),$ then $\mathcal{M}_{[k]}$ satisfies $\left(\sum_{i=1}^{k} \epsilon_{i}, \sum_{i=1}^{k} \delta_{i}\right)$-differential privacy.
\end{theorem}
\begin{proof}
Please refer to the proof of Theorem 3.16 in Dwork and Roth~\cite{dwork2014algorithmic}
\end{proof}
The  composition theorem is a standard way to obtain privacy guarantees for repeated application of differentially private algorithms. Besides the composition theorem, $(\epsilon,\delta)$-differential privacy also has the following post-processing property: 
\begin{proposition}
[Post-processing~\cite{dwork2014algorithmic}]
\label{post-processing}
Let $\mathcal{M}: \mathcal{D} \rightarrow R$ be a randomized algorithm that is $(\epsilon, \delta)$-differentially private. If $f: R \rightarrow R^{\prime}$ is an arbitrary randomized or deterministic mapping. Then $f \circ \mathcal{M}: \mathcal{D} \rightarrow R^{\prime}$ satisfies $(\epsilon, \delta)$-differential privacy.
\end{proposition}
\begin{proof}
Please refer to the proof of Proposition 2.1 in Dwork and Roth~\cite{dwork2014algorithmic}.
\end{proof}

The post-processing property ensures that the computation results of a differentially private mechanism can be safely released because
any post-processing computation of originally $(\epsilon,\delta)$-differentially private algorithm will also be $(\epsilon,\delta)$-differentially private. The composition theorem and post-processing property make $(\epsilon,\delta)$-differential privacy applicable to analyze/design complex differentially private algorithms. 

\subsection{Bagging}
Ensemble learning~\cite{dietterich2000ensemble} tries to combine the base models produced by several learners into an ensemble that performs better than the original base learners.
Bagging is a popular ensemble learning framework~\cite{breiman1996bagging}, which we formally define as follows: 

\begin{definition}[Bagging (\textbf{B}ootstrap \textbf{Agg}regat\textbf{ing})~\cite{breiman1996bagging}] Given a training dataset $D$ of size $n$, Bagging generates $N$ subsamples $D_{i}~ (i=1,2,\cdots,N)$. Each subsample $D_{i}$ randomly samples $k$ examples from $D$ {with or without replacement}. Then, Bagging trains a base model on each subsample $D_{i}$ using a base learner. When predicting the label for a testing example, Bagging takes majority vote among the base models.
\end{definition}

If Bagging creates a subsample by randomly sampling $k$ examples from the training dataset \emph{with replacement}, some examples in the training dataset may be chosen multiple times. If Bagging creates a subsample by randomly sampling $k$ examples from the training dataset \emph{without replacement}, an example in the training dataset may be selected at most once. In the following parts of this paper, we distinguish these two methods as Bagging \emph{with replacement} and Bagging \emph{without replacement}, respectively. In Section \ref{sec:methodology}, our major theoretical results show that Bagging's intrinsic randomness brought by its re-sampling process already satisfies $(\epsilon,\delta)$-differential privacy.

\subsection{Differentially Private Machine Learning}
Many studies~\cite{hamm2016learning,abadi2016deep,papernot2016semi,papernot2018scalable,jordon2018pate,xie2018differentially,chaudhuri2011differentially,kifer2012private,song2013stochastic,bassily2014private,wang2017differentially,jordon2018pate} have extended $(\epsilon,\delta)$-differential privacy to machine learning. Generally speaking, most studies satisfy differential privacy by introducing additive-noise mechanisms. Specifically, existing methods introduce noise/randomness to the training data, the gradients when stochastic gradient descent is used to learn a model, the loss function, and/or the model itself. For instance, Chaudhuri et al.~\cite{chaudhuri2011differentially} proposed to add noise to the loss function and then minimize the noisy loss function using a standard optimization method. Kifer et al.~\cite{kifer2012private} improved the utility of such loss function based perturbation method. Several other methods~\cite{song2013stochastic,bassily2014private,abadi2016deep,wang2017differentially} proposed to add noise to the gradient in each iteration of gradient descent or stochastic gradient descent. For instance, Abadi et al.~\cite{abadi2016deep} proposed  DPSGD, which introduces well-calibrated Gaussian noise to the gradient computed from a random batch of the training data in each iteration when using stochastic gradient descent to learn a model. Moreover, they
proposed moments accountant, which is a stronger accounting method to track the privacy loss for adding Gaussian noise than the standard composition theorem.  
Jordan et al.~\cite{jordon2018pate} proposed a method for the generator in generative adversarial networks~\cite{goodfellow2014generative} to generate synthetic data for training, which provides privacy guarantee for the original training dataset. Papernot et al.~\cite{papernot2016semi,papernot2018scalable} developed the PATE~\cite{papernot2016semi,papernot2018scalable} framework, which trains multiple teacher models on pre-defined disjoint chunks of the sensitive training data and distills the teacher models to a student model using public non-sensitive data in a privacy-preserving way. The student model satisfies $(\epsilon,\delta)$-differential privacy. Jordan et al.~\cite{jordon2019differentially}  introduced a variant of PATE to improve the student model's accuracy by dividing the sensitive data several times (rather than just once in PATE) and learning teacher models on each chunk within each division. Note that PATE and its variant~\cite{jordon2019differentially} divide the sensitive training data to chunks, which is different from Bagging that randomly creates subsamples. Moreover, our results essentially show that if PATE trains the teacher models using randomly created subsamples, then the teacher models (they can be treated as base models in Bagging) already satisfy  $(\epsilon,\delta)$-differential privacy.


\section{$(\epsilon,\delta)$-Differential Privacy of Bagging}
\label{sec:methodology}

In this section, we first intuitively explain why the randomness of the re-sampling process in Bagging may satisfy differential privacy. Then, we formally derive the $(\epsilon,\delta)$-differential privacy of Bagging in different cases and prove the tightness of our derived privacy guarantee bounds. Due to the space limitation, we put our detailed proofs in the Supplementary Material.  

Roughly speaking, a machine learning method satisfies differential privacy if the learnt model does not change much when adding or removing one example in the training data. Suppose Bagging randomly samples $k$ examples with replacement from the training dataset to create one subsample and trains one base model. When the size of the training dataset is large, Bagging is unlikely to sample the example, which is added or removed from the original training dataset. Specifically, when the size of the training dataset is $n$, Bagging {with replacement} would not select that  added or removed example with probability of $\left(\frac{n-1}{n}\right)^{k}$ when creating a subsample. Therefore, adding or removing one example in the training dataset is unlikely to substantially affect the base model, making Bagging differentially private. In the following, we formally analyze the $(\epsilon,\delta)$-differential privacy guarantees of Bagging.  Table \ref{epsilon,delta} summarizes our derived $(\epsilon,\delta)$-differential privacy guarantees  for Bagging with and without replacement. 

\begin{table}[!t]
 \renewcommand\arraystretch{1.5} 
 \captionsetup{justification=centering}
  \caption{Our derived tight $(\epsilon,\delta)$-differential privacy guarantees for Bagging.  $n$ is the training dataset size, $k$ is the subsample size, and $N$ is the number of base models. }
  \label{epsilon,delta}
  \centering
  \begin{tabular}{ccc}
    \toprule
      &  $\epsilon$ & $\delta$   \\
    \midrule
    Bagging \emph{with replacement} & $N\cdot k \cdot \ln{\frac{n+1}{n }}$ &$1- (\frac{n-1}{n})^{N\cdot k}$ \\
    Bagging \emph{without replacement} & $\ln{\frac{n+1}{n+1-N\cdot k}}$ & $\frac{N\cdot k}{n}$\\
    \bottomrule
  \end{tabular}
\end{table}

\begin{theorem}
[$(\epsilon,\delta)$-differential privacy of Bagging {with replacement} when $N=1$]
\label{bagging_dp_theorem_2.2}
 Given a training dataset of size $n$ and an arbitrary base learner, Bagging with replacement achieves $(k \cdot \ln{\frac{n+1}{n}},1- (\frac{n-1}{n})^k)$-differential privacy when training one base model, where $k$ is the subsample size.
\end{theorem}
\begin{proof}
Please refer to Appendix~\ref{proof_of_bagging_dp_theorem_2.2} in Supplementary Material. 
\end{proof}

Note that our results are applicable to any base learner. One way to obtain the differential privacy guarantee of Bagging when training $N$ base models is to apply the standard composition theorem in Theorem~\ref{composition theorem}. In particular, based on the standard composition theorem of $(\epsilon,\delta)$-differential privacy,  Bagging {with replacement} achieves $\left(N \cdot k \cdot \ln{\frac{n+1}{n}},N \cdot \left( 1- (\frac{n-1}{n})^k\right)\right)$-differential privacy when training $N$ base models. However,  this differential privacy guarantee from the standard composition theorem is loose. In the following theorem, we show that Bagging with $N$ base models achieves better differential privacy guarantees than that indicated by the standard composition theorem. 

\begin{theorem}
[$(\epsilon,\delta)$-differential privacy of Bagging {with replacement} when $N \textgreater 1$]
\label{theorem-3}
 Given a training dataset of size $n$ and an arbitrary base learner, Bagging with replacement achieves $(N\cdot k \cdot \ln{\frac{n+1}{n }},1- (\frac{n-1}{n})^{N\cdot k})$-differential privacy when training $N$ base models, where $k$ is the subsample size.
\end{theorem}
\begin{proof}
We first sample $N \cdot k$ examples from the training dataset uniformly at randomly {with replacement}. Based on the proof of Theorem~\ref{bagging_dp_theorem_2.2},  the  sampled $N \cdot k$ examples achieve $\left(N\cdot k \cdot \ln{\frac{n + 1}{n}}, 1 - (\frac{n-1}{n})^{N\cdot k}\right)$-differential privacy. Then, we can evenly divide the $N\cdot k$ examples to $N$ subsamples and train $N$ base models. Note that we can view training the $N$ base models as  post-processing of the $N \cdot k$ examples, which does not incur extra privacy loss based on Proposition~\ref{post-processing}. Therefore, the entire process achieves  $(N\cdot k \cdot \ln{\frac{n + 1}{n}}, 1 - (\frac{n-1}{n})^{N\cdot k})$-differential privacy.
\end{proof}

If Bagging randomly samples $k$ examples from the training dataset {without replacement} to create each subsample, then Bagging has the following differential privacy guarantee: 
\begin{theorem}
[$(\epsilon,\delta)$-differential privacy of Bagging {without replacement}]
\label{theorem4}
 Given a training dataset of size $n$ and an arbitrary base learner,  Bagging without replacement achieves $(\ln{\frac{n+1}{n+1-N\cdot k}},$\\$ \frac{N\cdot k}{n})$-differential privacy when training $N$ base models, where $k$ is the subsample size.
\end{theorem}
\begin{proof}
Please see Appendix \ref{appendix-C} in Supplementary Material. 
\end{proof}

Next, we show that our derived privacy guarantees of Bagging are tight if no extra assumptions are made on the base learner. More specifically, if no assumptions on the base learner are made, it is impossible to derive a $\delta$ that is smaller than ours for Bagging.

\begin{theorem}
[Tightness of $\delta$ for Bagging with replacement]\label{tightness-theorem-3}
For any $\delta<1-\left(\frac{n-1}{n}\right)^{N \cdot k}$, there exists a base learner such that Bagging with replacement cannot satisfy $(\epsilon, \delta)$-differential privacy for any $\epsilon$.
\end{theorem}
\begin{proof}
Our proof is based on constructing a counter-example base leaner that Bagging with replacement cannot achieve $(\epsilon,\delta)$-differential privacy for any $\epsilon$ when  $\delta< 1 - (\frac{n-1}{n})^{N \cdot k}$. See Appendix \ref{appendix-D} in Supplementary Material for more details.
\end{proof}

\begin{theorem}
[Tightness of $\delta$ for Bagging without replacement]
\label{tightness-theorem-4}
For any $\delta<\frac{N \cdot k}{n}$, there exists a base learner such that Bagging without replacement cannot satisfy $(\epsilon, \delta)$-differential privacy for any $\epsilon$.
\end{theorem}

\begin{proof}
Please see Appendix \ref{appendix-E} in Supplementary Material.
\end{proof}


\section{Evaluation}
We compare Bagging with DPSGD~\cite{abadi2016deep} and PATE~\cite{papernot2018scalable} in different scenarios. 
We use the open-source implementations of DPSGD and PATE from their authors. 

\subsection{Experimental Setup}
\myparatight{Two cases} Privacy-preserving machine learning aims to train a model while protecting the privacy of a sensitive training dataset. Depending on whether we have access to a public non-sensitive dataset, we consider the following two cases. 
\begin{itemize}
    \item {\bf Case I:} 
    \textbf{No access to a public non-sensitive dataset.} In this case, we only have access to the sensitive training dataset. In this case, PATE is not applicable. Therefore, we compare Bagging and DPSGD in this case.
    \item {\bf Case II:}
    \textbf{Access to a public non-sensitive dataset.} In this case, we have access to a public non-sensitive dataset other than the sensitive training dataset. For instance, the sensitive training dataset could be CIFAR10 while the public non-sensitive dataset could be ImageNet. Therefore, we can leverage transfer learning to distill knowledge from the public non-sensitive dataset to boost the accuracy of the privacy-preserving model trained on the sensitive training dataset. In particular, we can first pretrain a model on the public dataset. Then, for DPSGD, we fine-tune the pretrained model on the sensitive training dataset using DPSGD. For Bagging, we fine-tune the pretrained model on subsamples of the sensitive training dataset to obtain the base models. For PATE, we train the teacher models and student model via fine tuning the pretrained model. Note that DPSGD and Bagging do not require the public non-sensitive dataset to have the same distribution as the sensitive training dataset. However, PATE further requires a public dataset that has the same distribution as the sensitive training dataset to train the student model. We call this public dataset \emph{same-distribution public dataset}. Therefore, PATE has stronger assumptions than DPSGD and Bagging.     
\end{itemize}

\myparatight{Datasets and models} We discuss our datasets and models for \textbf{Case I} and \textbf{Case II} separately. 
\begin{itemize}

\item {\bf Case I:} We adopt MNIST~\cite{lecun2010mnist} and CIFAR10~\cite{krizhevsky2009learning} as the sensitive training datasets. On MNIST, we use a simple convolutional neural network with two convolutional layers, each followed by a pooling layer, and a fully connected layer (Table~\ref{MNIST-structure} in Supplemental Material shows the details). For CIFAR10, we adopt the VGG16~\cite{simonyan2014very} architecture. 
\item {\bf Case II:} We still adopt  MNIST and CIFAR10 as the sensitive training datasets. When MNIST is the sensitive training dataset, we adopt Fashion-MNIST~\cite{xiao2017fashion} as the public non-sensitive dataset. When CIFAR10 is the sensitive training dataset, we assume ImageNet~\cite{deng2009imagenet} is the public non-sensitive dataset. PATE further requires a small same-distribution public dataset. For this purpose, we select the first 1,000 testing examples of  MNIST (or CIFAR10) as the same-distribution public dataset when training PATE's student models on MNIST (or CIFAR10). Note that DPSGD, Bagging, and PATE are all evaluated on the  remaining 9,000 testing examples of MNIST (or CIFAR10) in \textbf{Case II}. On Fashion-MNIST, we pretrained a convolutional neural network, which has the same architecture as the model in \textbf{Case I} for MNIST. Moreover, we adopt the pretrained VGG16 model\footnote{https://github.com/keras-team/keras-applications/blob/master/keras\_applications/vgg16.py}~\cite{simonyan2014very} for the ImageNet dataset. DPSGD, PATE, and Bagging fine-tune these pretrained models on the sensitive training dataset following the standard fine-tuning procedure. In particular, we replace the last fully connected layer of a pretrained model as a new one that has the same number of classes as the sensitive training dataset. We then fine-tune the model using a learning rate that is 10 times smaller than that when training from scratch. 
\end{itemize}

\myparatight{Parameter settings}  We set training epochs=100 for both Bagging and DPSGD in both \textbf{Case I} and \textbf{Case II}. We adopt Bagging {with replacement} and $N=1$ as the default setting in \textbf{Case I }and \textbf{Case II}. For PATE~\cite{papernot2018scalable}, we set the number of teachers to be 250, and both teacher and student models  are trained for 1,000 epochs. Bagging, DPSGD, and PATE have different ways to control $\epsilon$ and $\delta$. Next, we describe how to set their parameters to achieve a target level of $\epsilon$ and $\delta$.

\begin{itemize}
    \item Bagging. Given the training dataset size $n$ and subsample size $k$, we can calculate the privacy budget $(\epsilon,\delta)$ of Bagging based on Table~\ref{epsilon,delta}.
    \item DPSGD~\cite{abadi2016deep}. We vary the standard deviation $\sigma$ of the Gaussian noise  used by DPSGD to achieve a target level of $\epsilon$ and $\delta$. 
    \item PATE~\cite{papernot2018scalable}.  For PATE with the Confident-GNMax aggregation mechanism~\cite{papernot2018scalable}, threshold T and noise parameter $\sigma_{1}$ are used for privately checking  consensus of the teachers' predictions. Gaussian noise standard deviation $\sigma_{2}$ is used for the usual max-of-Gaussian~\cite{papernot2018scalable}. Following the authors of PATE, we set T$=200$, $\sigma_{1}=150$, and $\sigma_{2}=40$. We then vary the  number of queries answered by the Confident-GNMax aggregator to achieve a target privacy budget $\epsilon$ and $\delta$. We found that PATE cannot reach the small $\epsilon$ we set for Bagging and DPSGD, so we  relax its $\epsilon$ to be around 1, which is much larger than the $\epsilon$ used by Bagging and DPSGD. In other words, we give additional advantages for PATE. 
\end{itemize}

\subsection{Experimental Results}
We first compare Bagging with DPSGD and PATE. Then, we evaluate different variants of Bagging.  

\myparatight{Comparison results in Case I}
Table \ref{MNIST-1} and \ref{CIFAR10-1} show the testing accuracies of DPSGD and Bagging (with replacement and $N=1$) for different privacy budgets in Case I on MNIST and CIFAR10, respectively. The column ``No Privacy'' corresponds to models without privacy guarantees. First,  we observe that Bagging achieves significantly higher testing accuracies than DPSGD under the same privacy budget. Second, increasing the subsample size $k$ in Bagging is equivalent to decreasing the Gaussian noise scale in DPSGD, which provides weaker privacy guarantees and trains models with higher accuracies. 

\begin{table}[!t]
  \caption{ Comparison results on MNIST in Case I.}
  \label{MNIST-1}
  \centering
  
  \begin{tabular}{lllllll}
    \toprule
    \multicolumn{2}{c}{Privacy Budget} &
   \multicolumn{3}{c}{Accuracy}     &     \multicolumn{2}{c}{Parameter Setting}           \\
    \cmidrule(r){1-2} \cmidrule(r){3-5}   \cmidrule(r){6-7}
     $\epsilon$ & $\delta$     & No Privacy &  DPSGD~\cite{abadi2016deep}    & Bagging  &  $\sigma$ in DPSGD & $k$ in Bagging   \\
    \midrule
    0.005&0.005& -&16.41\%&\textbf{90.66\%} &200&300\\
    0.0083&0.0083&- &19.76\%&\textbf{93.81\%}& 150&500\\
    0.017&0.017&- &39.73\%&\textbf{94.60\%}&100&1000\\
    0.083&0.08& -&87.91\%&\textbf{97.68\%} &19&5000\\
    0.167&0.154&- &91.95\%&\textbf{98.28\%} &8&10000\\
    \midrule
    -&-&99.1\%&-&-&-&-\\
    \bottomrule
  \end{tabular}
\end{table}

\begin{table}[!t]
  \caption{Comparison results on CIFAR10 in Case I.}
  \label{CIFAR10-1}
  \centering
  \begin{tabular}{lllllll}
    \toprule
    \multicolumn{2}{c}{Privacy Budget} &
   \multicolumn{3}{c}{Accuracy}     &     \multicolumn{2}{c}{Parameter Setting}           \\
    \cmidrule(r){1-2} \cmidrule(r){3-5}   \cmidrule(r){6-7}
     $\epsilon$ & $\delta$     & No Privacy &  DPSGD~\cite{abadi2016deep}    & Bagging  &  $\sigma$ in DPSGD & $k$ in Bagging   \\
    \midrule
    0.02&0.02&-&12.74\%&\textbf{41.63\%} &72&1000\\
    0.1&0.095&-&15.78\%&\textbf{59.76\%} &11.5&5000\\
    0.2&0.181&- &27.86\%&\textbf{65.25\%}&4.9&10000\\
    0.4&0.33&-&40.96\%&\textbf{70.75\%} &2.18&20000\\
    0.6&0.45&-&46.53\%&\textbf{73.95\%} &1.4&30000\\
    \midrule
    -&-&80.82\%&-&-&-&-\\
    \bottomrule
  \end{tabular}
\end{table}

\myparatight{Comparison results in Case II}
Table \ref{MNIST-2} and \ref{CIFAR10-2} respectively show the testing accuracies of DPSGD, PATE, and Bagging for different privacy budgets in \textbf{Case II} on MNIST and CIFAR10, respectively. We have two observations. First, Bagging achieves significantly higher testing accuracies than DPSGD and PATE, even if PATE has weaker privacy guarantees. Second,  via comparing Table \ref{CIFAR10-1} and Table \ref{CIFAR10-2}, we observe that Bagging achieves better accuracies in \textbf{Case II} than \textbf{Case I} for CIFAR10 when the public non-sensitive dataset is ImageNet, which means that transferring knowledge from a public non-sensitive dataset does improve accuracy of the model trained on the sensitive training dataset. DPSGD also achieves higher accuracies in \textbf{Case II} for CIFAR10 when the privacy budgets are larger than some threshold (e.g., 0.02). However, we didn't observe such accuracy improvement for MNIST in \textbf{Case II} when the public non-sensitive dataset is Fashion-MNIST. In fact, based on Table~\ref{MNIST-1} and~\ref{MNIST-2}, testing accuracies of DPSGD and Bagging may even decrease in \textbf{Case II}. We suspect the reason may be that the pretrained model for Fashion-MNIST is much simpler than that for ImageNet, which does not extract meaningful features.

\begin{table}[!t]
  \caption{Comparison results on MNIST in Case II.}
  \label{MNIST-2}
  \centering
  
  \begin{tabular}{lllll}
    \toprule
    Method &  $\epsilon$ &  $\delta$ &
   Accuracy     &     Parameter Setting        \\
       \midrule
    No Privacy & - & - & 98.83\% & - \\
    \midrule
    DPSGD~\cite{abadi2016deep}&0.008&0.008&11.17\%& $\sigma$=150\\
    PATE~\cite{papernot2018scalable}& 0.57 & 0.008 & 74.37\% & Quiries Answered by Aggregator= 100\\
    Bagging & 0.008 &0.008&\textbf{90.20\%}& $k=500$\\
    \midrule
    DPSGD~\cite{abadi2016deep}&0.017&0.017&14.46\%& $\sigma$=100\\
    PATE~\cite{papernot2018scalable}& 0.66 & 0.017 &  78.42\% & Quiries Answered by Aggregator = 150 \\
    Bagging & 0.017 &0.017&\textbf{93.59\%}& $k=1000$\\
    \midrule
    DPSGD~\cite{abadi2016deep}&0.08&0.08&81.43\% & $\sigma$=19\\
    PATE~\cite{papernot2018scalable}& 0.63 & 0.08 & 83.31\% & Quiries Answered by Aggregator = 180 \\
    Bagging & 0.08 &0.08&\textbf{96.87\%}& $k=5000$\\
    \bottomrule
  \end{tabular}
\end{table}

\begin{table}[!t]
  \caption{Comparison results on CIFAR10 in Case II.}
  \label{CIFAR10-2}
  \centering
  
  \begin{tabular}{lllll}
    \toprule
    Method &  $\epsilon$ &  $\delta$ &
   Accuracy     &     Parameter Setting        \\
       \midrule
    No Privacy & - & - & 87.90\% & - \\
    \midrule
    DPSGD~\cite{abadi2016deep}&0.02&0.02&11.97\%& $\sigma$=72     \\
    PATE~\cite{papernot2018scalable}& 1.12 & 0.02 & 37.8\% & Quiries Answered by Aggregator = 100 \\
    Bagging & 0.02 &0.02&\textbf{62.22\%}& $k=1000$\\
    \midrule
    DPSGD~\cite{abadi2016deep}& 0.1 &0.095&21.08\% & $\sigma$=11.5\\
    PATE~\cite{papernot2018scalable}& 1.02 & 0.095 & 40.86\% &Quiries Answered by Aggregator= 130 \\
    Bagging & 0.1 &0.095&\textbf{75.67\%}& $k=5000$\\
    \midrule
    DPSGD~\cite{abadi2016deep}&0.2&0.18&30.63\% & $\sigma$=4.9\\
    PATE~\cite{papernot2018scalable}& 1.03 & 0.18 & 42.55\% &Quiries Answered by Aggregator = 170 \\
    Bagging & 0.2 &0.18&\textbf{79.55\%}& $k=10000$\\
    \bottomrule
  \end{tabular}
\end{table}

\myparatight{Impact of $k$ and $N$ on Bagging}
Figure \ref{fig:1} shows the testing accuracy of Bagging in \textbf{Case II} as we train more base models. 
We fix $N \cdot k$ in each curve in the graphs, so each curve has the same privacy budget independent of the number of base models. We observe that, given the same privacy budget, Bagging has lower accuracies when training more base models.  The reason is that, given the same privacy budget, training more base models means that each base model is trained using less examples and thus less accurate. 

\begin{figure*}[!t]
\centering
\subfloat[MNIST]{\includegraphics[width=0.43\textwidth]{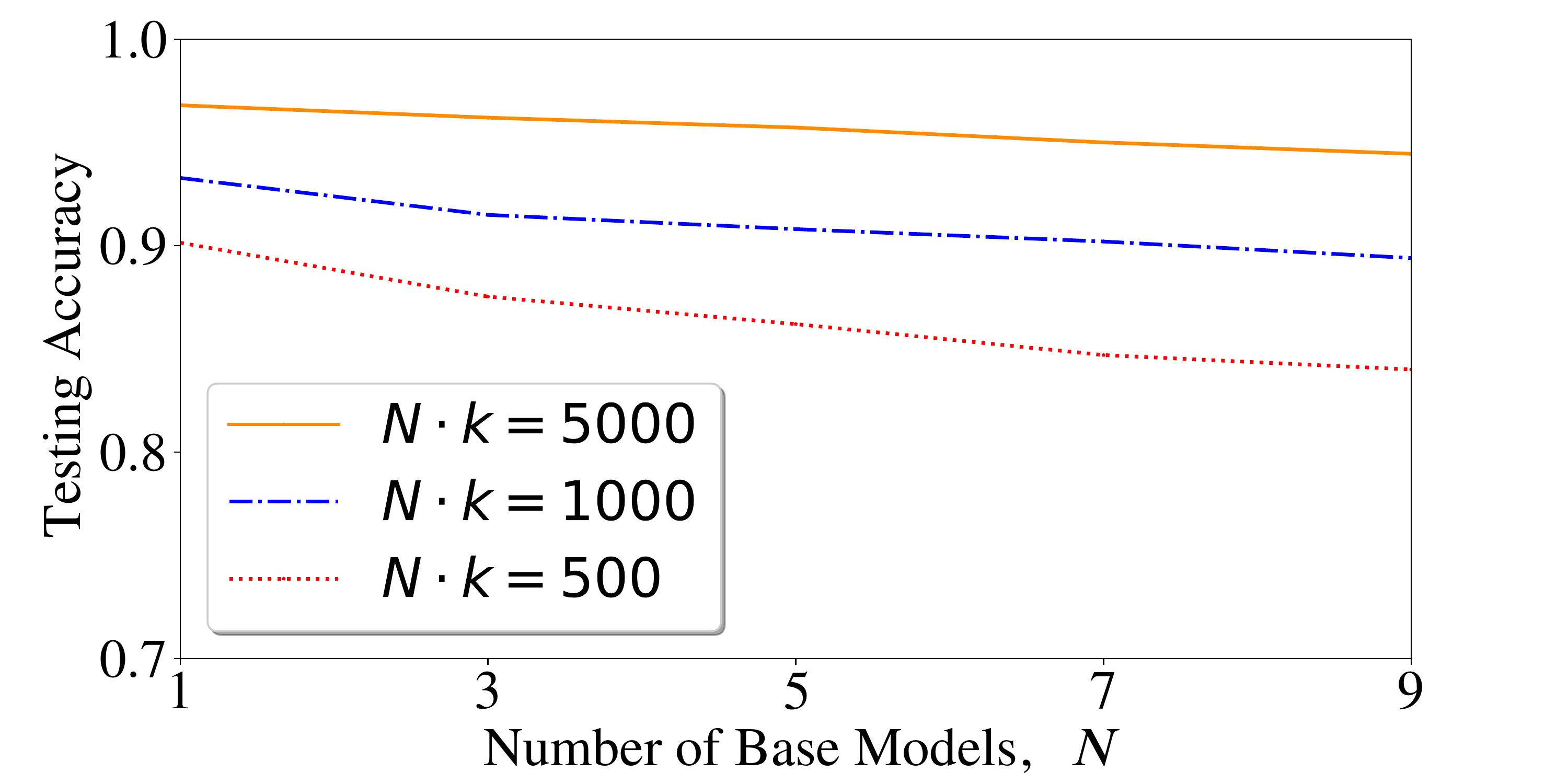}}
\subfloat[CIFAR10]{\includegraphics[width=0.43\textwidth]{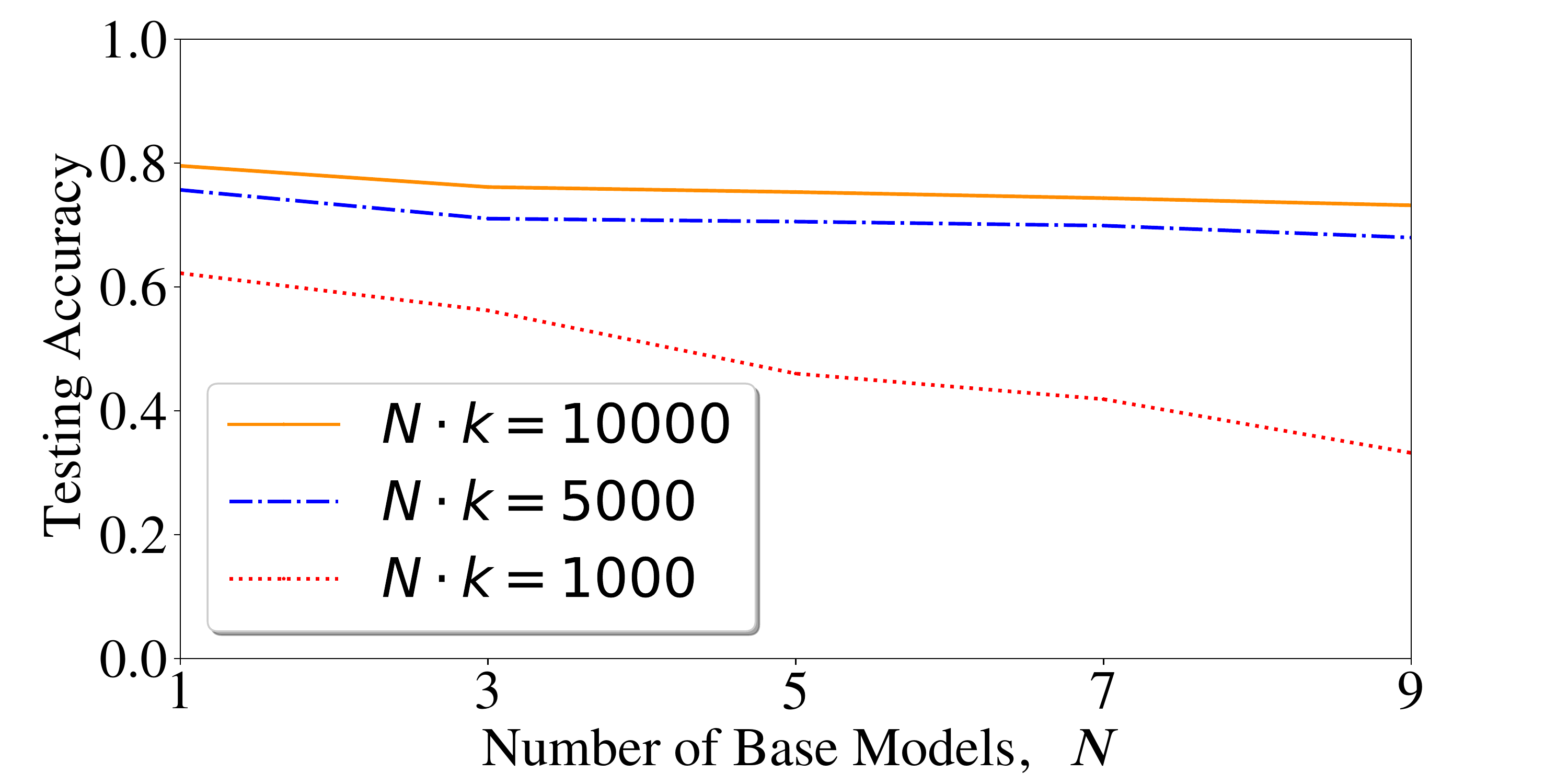}}
\caption{Testing accuracy of Bagging when training $N$ base models, where the privacy budget $N\cdot k$ is fixed.}
\label{fig:1}
\end{figure*}

\myparatight{Bagging with vs. without replacement}
Figure \ref{fig:2} shows the testing accuracy of Bagging with vs. without replacement in \textbf{Case II} as we vary the subsample size $k$, where $N=1$. The same subsample size $k$ has very close privacy budgets $(\epsilon,\delta)$ for Bagging with and without replacement, according to Table~\ref{epsilon,delta}. Our results show that  with or without replacement has negligible  impact on Bagging. 

\begin{figure*}[!t]
\centering
\subfloat[MNIST]{\includegraphics[width=0.43\textwidth]{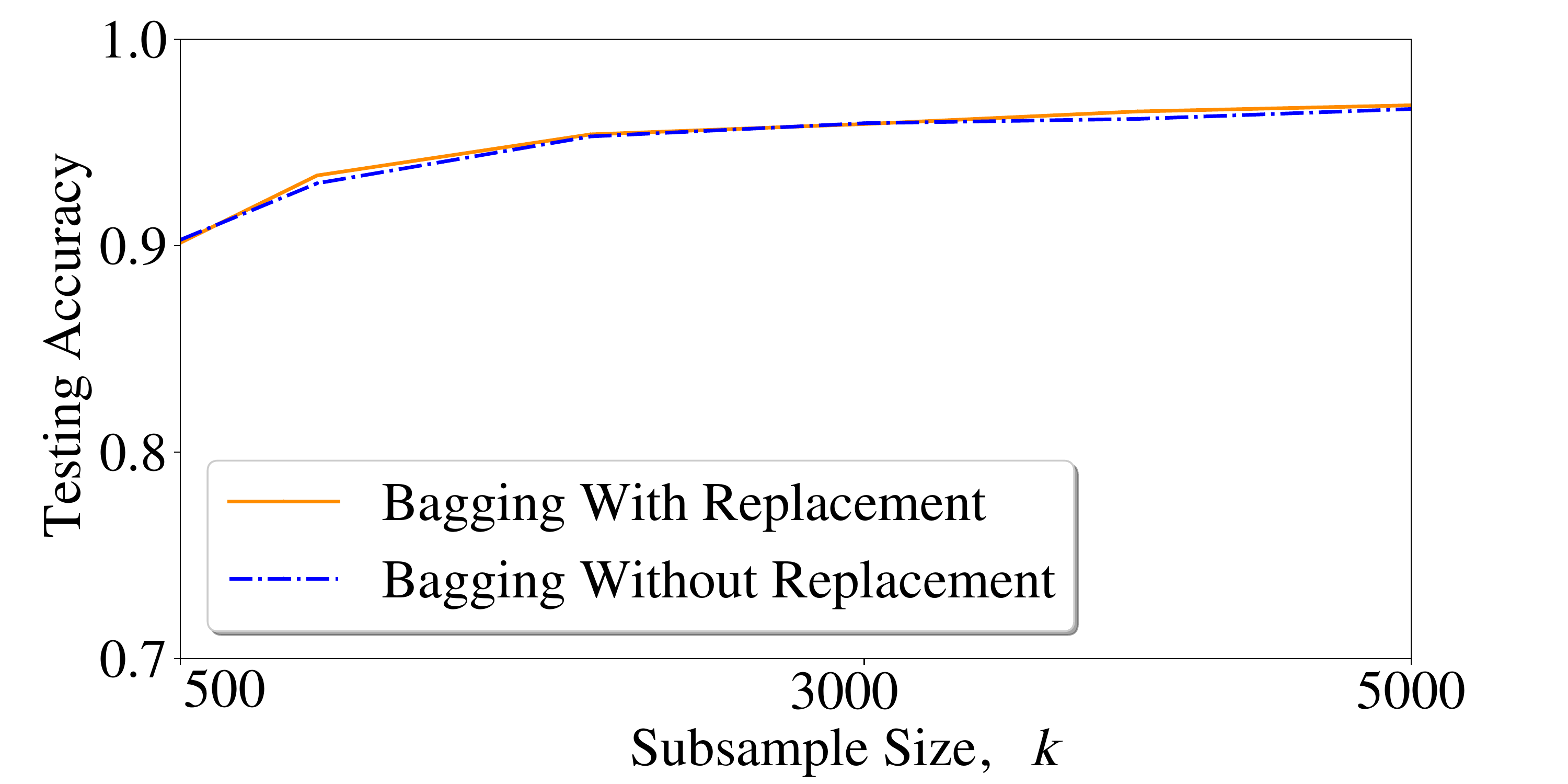}}
\subfloat[CIFAR10]{\includegraphics[width=0.43\textwidth]{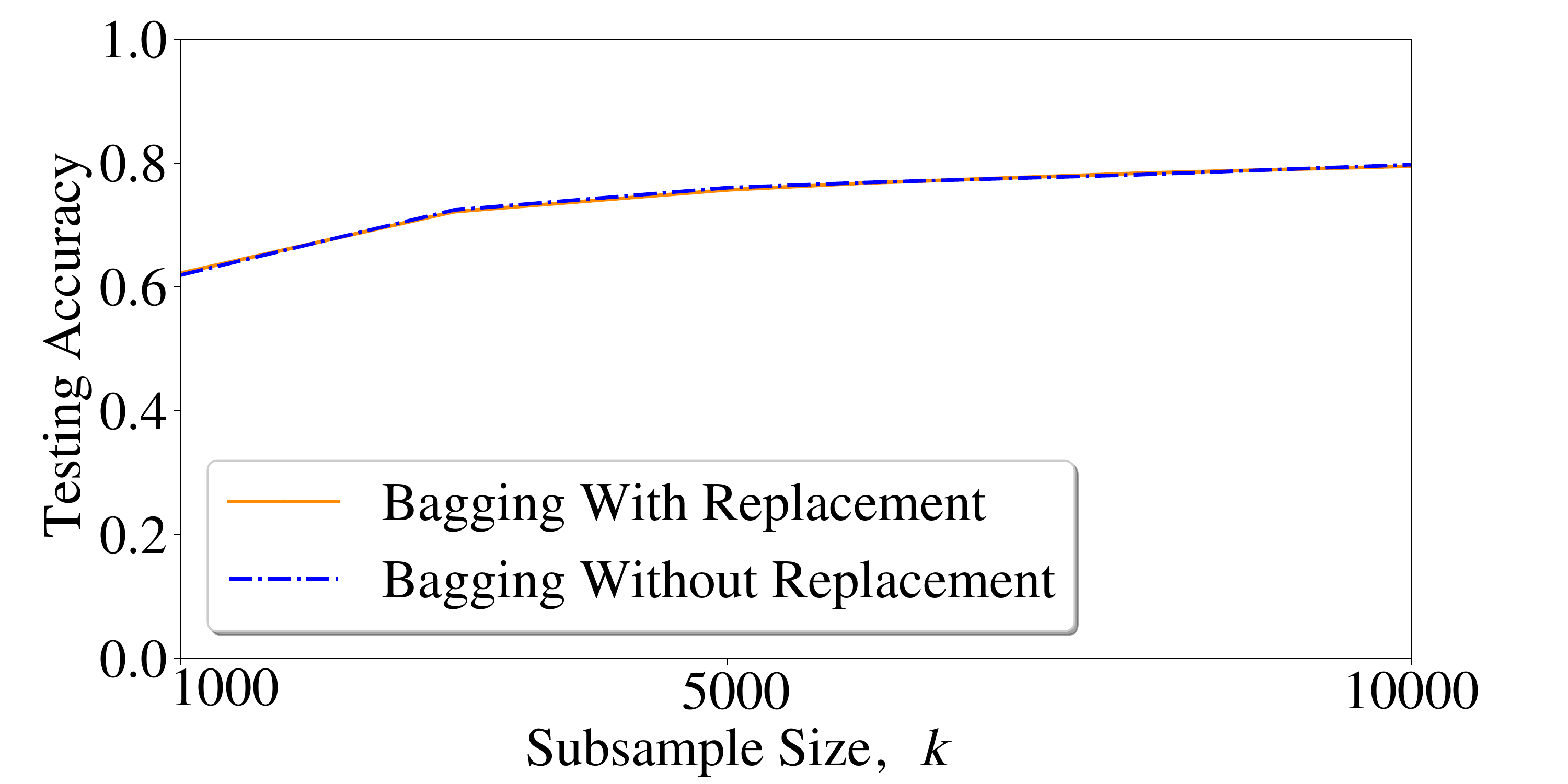}}
\caption{Bagging with replacement vs. Bagging without replacement.}
\label{fig:2}
\end{figure*}


\section{Conclusion}
\label{sec:conclusion}
In this work, we study the intrinsic $(\epsilon,\delta)$-differential privacy of Bagging. Our major theoretical results show that Bagging's intrinsic randomness from  subsampling  already makes Bagging differentially private without the needs of additional noise. We derive the $(\epsilon,\delta)$-differential privacy guarantees for Bagging with and without replacement. Moreover, we  prove that if no assumptions about the base learner are made, our derived privacy guarantees are tight. We empirically evaluate Bagging on MNIST and CIFAR10. Our experimental results demonstrate that {Bagging} achieves significantly higher accuracies than  state-of-the-art differentially private machine learning methods with the same privacy budget.


\bibliographystyle{plain}
\bibliography{refs}

\begin{thebibliography}{10}

\bibitem{abadi2016deep}
Martin Abadi, Andy Chu, Ian Goodfellow, H~Brendan McMahan, Ilya Mironov, Kunal
  Talwar, and Li~Zhang.
\newblock Deep learning with differential privacy.
\newblock In {\em Proceedings of the 2016 ACM SIGSAC Conference on Computer and
  Communications Security}, pages 308--318, 2016.

\bibitem{bassily2014private}
Raef Bassily, Adam Smith, and Abhradeep Thakurta.
\newblock Private empirical risk minimization: Efficient algorithms and tight
  error bounds.
\newblock In {\em 2014 IEEE 55th Annual Symposium on Foundations of Computer
  Science}, pages 464--473. IEEE, 2014.

\bibitem{breiman1996bagging}
Leo Breiman.
\newblock Bagging predictors.
\newblock {\em Machine learning}, 24(2):123--140, 1996.

\bibitem{chaudhuri2011differentially}
Kamalika Chaudhuri, Claire Monteleoni, and Anand~D Sarwate.
\newblock Differentially private empirical risk minimization.
\newblock {\em Journal of Machine Learning Research}, 12(Mar):1069--1109, 2011.

\bibitem{deng2009imagenet}
Jia Deng, Wei Dong, Richard Socher, Li-Jia Li, Kai Li, and Li~Fei-Fei.
\newblock Imagenet: A large-scale hierarchical image database.
\newblock In {\em 2009 IEEE conference on computer vision and pattern
  recognition}, pages 248--255. Ieee, 2009.

\bibitem{dietterich2000ensemble}
Thomas~G Dietterich.
\newblock Ensemble methods in machine learning.
\newblock In {\em International workshop on multiple classifier systems}, pages
  1--15. Springer, 2000.

\bibitem{dwork2014algorithmic}
Cynthia Dwork and Aaron Roth.
\newblock The algorithmic foundations of differential privacy.
\newblock {\em Foundations and Trends{\textregistered} in Theoretical Computer
  Science}, 9(3--4):211--407, 2014.

\bibitem{fredrikson2015model}
Matt Fredrikson, Somesh Jha, and Thomas Ristenpart.
\newblock Model inversion attacks that exploit confidence information and basic
  countermeasures.
\newblock In {\em CCS}, 2015.

\bibitem{goodfellow2014generative}
Ian Goodfellow, Jean Pouget-Abadie, Mehdi Mirza, Bing Xu, David Warde-Farley,
  Sherjil Ozair, Aaron Courville, and Yoshua Bengio.
\newblock Generative adversarial nets.
\newblock In {\em Advances in neural information processing systems}, pages
  2672--2680, 2014.

\bibitem{hamm2016learning}
Jihun Hamm, Yingjun Cao, and Mikhail Belkin.
\newblock Learning privately from multiparty data.
\newblock In {\em International Conference on Machine Learning}, pages
  555--563, 2016.

\bibitem{jordon2018pate}
James Jordon, Jinsung Yoon, and Mihaela van~der Schaar.
\newblock Pate-gan: Generating synthetic data with differential privacy
  guarantees.
\newblock 2018.

\bibitem{jordon2019differentially}
James Jordon, Jinsung Yoon, and Mihaela van~der Schaar.
\newblock Differentially private bagging: Improved utility and cheaper privacy
  than subsample-and-aggregate.
\newblock In {\em Advances in Neural Information Processing Systems}, pages
  4325--4334, 2019.

\bibitem{kifer2012private}
Daniel Kifer, Adam Smith, and Abhradeep Thakurta.
\newblock Private convex empirical risk minimization and high-dimensional
  regression.
\newblock In {\em Conference on Learning Theory}, pages 25--1, 2012.

\bibitem{krizhevsky2009learning}
Alex Krizhevsky, Geoffrey Hinton, et~al.
\newblock Learning multiple layers of features from tiny images.
\newblock 2009.

\bibitem{lecun2010mnist}
Yann LeCun, Corinna Cortes, and CJ~Burges.
\newblock Mnist handwritten digit database.
\newblock 2010.

\bibitem{papernot2016semi}
Nicolas Papernot, Mart{\'\i}n Abadi, Ulfar Erlingsson, Ian Goodfellow, and
  Kunal Talwar.
\newblock Semi-supervised knowledge transfer for deep learning from private
  training data.
\newblock {\em arXiv preprint arXiv:1610.05755}, 2016.

\bibitem{papernot2018scalable}
Nicolas Papernot, Shuang Song, Ilya Mironov, Ananth Raghunathan, Kunal Talwar,
  and {\'U}lfar Erlingsson.
\newblock Scalable private learning with pate.
\newblock {\em arXiv preprint arXiv:1802.08908}, 2018.

\bibitem{shokri2017membership}
Reza Shokri, Marco Stronati, Congzheng Song, and Vitaly Shmatikov.
\newblock Membership inference attacks against machine learning models.
\newblock In {\em 2017 IEEE Symposium on Security and Privacy (SP)}, pages
  3--18. IEEE, 2017.

\bibitem{simonyan2014very}
Karen Simonyan and Andrew Zisserman.
\newblock Very deep convolutional networks for large-scale image recognition.
\newblock {\em arXiv preprint arXiv:1409.1556}, 2014.

\bibitem{song2013stochastic}
Shuang Song, Kamalika Chaudhuri, and Anand~D Sarwate.
\newblock Stochastic gradient descent with differentially private updates.
\newblock In {\em 2013 IEEE Global Conference on Signal and Information
  Processing}, pages 245--248. IEEE, 2013.

\bibitem{voigt2017eu}
Paul Voigt and Axel Von~dem Bussche.
\newblock The eu general data protection regulation (gdpr).
\newblock {\em A Practical Guide, 1st Ed., Cham: Springer International
  Publishing}, 2017.

\bibitem{wang2017differentially}
Di~Wang, Minwei Ye, and Jinhui Xu.
\newblock Differentially private empirical risk minimization revisited: Faster
  and more general.
\newblock In {\em Advances in Neural Information Processing Systems}, pages
  2722--2731, 2017.

\bibitem{xiao2017fashion}
Han Xiao, Kashif Rasul, and Roland Vollgraf.
\newblock Fashion-mnist: a novel image dataset for benchmarking machine
  learning algorithms.
\newblock {\em arXiv preprint arXiv:1708.07747}, 2017.

\bibitem{xie2018differentially}
Liyang Xie, Kaixiang Lin, Shu Wang, Fei Wang, and Jiayu Zhou.
\newblock Differentially private generative adversarial network.
\newblock {\em arXiv preprint arXiv:1802.06739}, 2018.

\end{thebibliography}


\newpage
\appendix
\section{Proof of Theorem~\ref{bagging_dp_theorem_2.2}}
\label{proof_of_bagging_dp_theorem_2.2}

We first show the process of subsampling in Bagging with replacement achieves $(k\cdot \ln{\frac{n + 1}{n}}, 1 - (\frac{n-1}{n})^k)$-differential privacy. Then we can view Bagging's training of the base model as post-processing of subsampling so that Bagging with replacement achieves $(k\cdot \ln{\frac{n + 1}{n}}, 1 - (\frac{n-1}{n})^k)$-differential privacy.

For simplicity, we use operator $b$ to denote the subsampling operation of Bagging with replacement. In particular, given the training dataset $D$ and its adjacent dataset $D^{\prime}$, $b(D)$ and $b(D^{\prime})$ denote a subsample of $k$ examples sampled from $D$ and $D^{\prime}$ uniformly at random with replacement, respectively. We use $\Phi$ to denote the joint domain of $b(D)$ and $b(D^{\prime})$. Specifically, each element $\phi \in \Phi$ is a subsample of $k$ examples. Given $b(D)$ and $b(D^{\prime})$, we denote the following three subsets of domain $\Phi$: 
\begin{align}
   &  \Gamma_{1} = \{\phi\in \Phi| \text{Pr}(b(D)=\phi )>0, \text{Pr}( b(D^{\prime})=\phi ) = 0\}, \\
   &  \Gamma_{2} = \{\phi\in \Phi| \text{Pr}(b(D)=\phi )=0, \text{Pr}( b(D^{\prime})=\phi ) > 0\}, \\
   &  \Gamma_{3} = \{\phi \in \Phi| \text{Pr}(b(D)=\phi )>0, \text{Pr}( b(D^{\prime})=\phi ) > 0\}.
\end{align}

Intuitively, $\Gamma_{1}$ includes all subsamples that can only be derived from $D$. $\Gamma_{2}$ includes all subsamples that can only be derived from $D^{\prime}$. $\Gamma_{3}$ includes all subsamples that can be derived from $D$ and $D^{\prime}$. We note that $\Gamma_{1}\cup \Gamma_{2} \cup \Gamma_{3}$ constitutes all possible subsamples of $b(D)$ and $b(D^{\prime})$. Since $D^{\prime}$ is an adjacent dataset of $D$ and $D$ has size $n$, $D^{\prime}$ has size $n \pm 1$. For simplicity, we use $|D|$ and $|D^{\prime}|$ to denote the size of $D$ and $D^{\prime}$, respectively. Note that there is just one example which is different between $D$ and $D^{\prime}$. Therefore, we have the following:
\begin{align}
    & |D\cap D^{\prime}| = n-1, \text{ if } |D^{\prime}|=n-1, \\
     & |D\cap D^{\prime}| = n, \text{ if } |D^{\prime}|=n+1.
\end{align}

In other words, the number of overlapping examples between $D$ and $D^{\prime}$ are $n-1$ and $n$ when the size of $D^{\prime}$ is $n-1$ and $n+1$, respectively. Since Bagging with replacement randomly samples $k$ training examples uniformly at random with replacement, we have the following: 
\begin{align}
 &   \text{Pr}(b(D)=\phi )= 
    \begin{cases}
     \frac{1}{(|D|)^{k}}, &\text{ if } \phi \in \Gamma_1 \cup \Gamma_3, \\
     0, &\text{ otherwise}. 
    \end{cases} \\
&    \text{Pr}(b(D^{\prime})=\phi )= 
    \begin{cases}
     \frac{1}{(|D^{\prime}|)^{k}}, &\text{ if } \phi \in \Gamma_2 \cup \Gamma_3, \\
     0, &\text{ otherwise}. 
    \end{cases}
\end{align}

Moreover, the number of elements in $\Gamma_1$, $\Gamma_2$, and $\Gamma_3$ can be computed as $|D|^{k}-|D\cap D^{\prime}|^k$, $|D^{\prime}|^{k}-|D\cap D^{\prime}|^k$, and $|D\cap D^{\prime}|^k$, respectively.  
Given the probability density and the size of $\Gamma_1$, $\Gamma_2$, and $\Gamma_3$, we have the following:
\begin{align}
 &   \text{Pr}(b(D)\in \Gamma_{3}) = \left(\frac{\mid D \cap D^{\prime}\mid}{\mid D\mid}\right)^{k},
    \text{Pr}(b(D)\in \Gamma_{1}) = 1 -  \left(\frac{\mid D \cap D^{\prime}\mid}{\mid D \mid}\right)^{k} \text{ and }
     \text{Pr}(b(D)\in \Gamma_{2}) = 0, \\
 &   \text{Pr}(b(D^{\prime})\in \Gamma_{3}) = \left(\frac{\mid D \cap D^{\prime}\mid}{\mid D^{\prime}\mid}\right)^{k},
    \text{Pr}(b(D^{\prime})\in \Gamma_{2}) = 1 -  \left(\frac{\mid D \cap D^{\prime}\mid}{\mid D^{\prime}\mid}\right)^{k} \text{ and }
     \text{Pr}(b(D^{\prime})\in \Gamma_{1}) = 0.
\end{align}

Recall our goal is to show Bagging with replacement achieves $(\epsilon,\delta)$-differential privacy. We will first find a $\delta$ and then determine the value of $\epsilon$ for the given $\delta$. In particular, for any $\Delta \subseteq \Gamma_1$ and $|D^{\prime}|=n-1$, we have the following: 
\begin{align}
    \text{Pr}(b(D)\in \Delta) \leq \text{Pr}(b(D)\in \Gamma_1) = 1-\left(\frac{n-1}{n}\right)^{k} \text{ and }  \text{Pr}(b(D^{\prime})\in \Delta) = 0. 
\end{align}

In other words, we have the following:
\begin{align}
 \forall \Delta \subseteq \Gamma_1,  \text{Pr}(b(D)\in \Delta) \leq e^{\epsilon}\cdot 0 + 1-\left(\frac{n-1}{n}\right)^{k}=e^{\epsilon}\cdot  \text{Pr}(b(D^{\prime})\in \Delta)+ \left(1-\left(\frac{n-1}{n}\right)^{k}\right).
\end{align}

Therefore, we let $\delta = 1-(\frac{n-1}{n})^{k}$. Next, we prove that for $\forall \Delta \subseteq \Phi$, i.e., $\Delta$ is an arbitrary subset of $\Phi$, we have the following:
\begin{align}
\label{appendix-13}
    \text{Pr}(b(D)\in \Delta) \leq \left(\frac{n + 1}{n}\right)^{k} \cdot \text{Pr}(b(D^{\prime})\in \Delta) + \left(1-\left(\frac{n-1}{n}\right)^{k}\right). 
\end{align}

We decompose $\Delta$ into three disjoint subsets, i.e., $\Delta =\Delta_{1}\cup \Delta_2 \cup \Delta_3$, where $\Delta_1 \subseteq \Gamma_1$, $\Delta_2 \subseteq \Gamma_2$, and $\Delta_3 \subseteq \Gamma_3$. We first consider $|D^{\prime}|=n+1$. Then, we have the following: 
\begin{align}
& \text{Pr}(b(D)\in \Delta) \\
=& \text{Pr}(b(D)\in \Delta_1) + \text{Pr}(b(D)\in \Delta_2) + \text{Pr}(b(D)\in \Delta_3) \\
=& \text{Pr}(b(D)\in \Delta_3) \\
=& \text{Pr}(b(D^{\prime})\in \Delta_3)\cdot \frac{|D^{\prime}|^k}{|D|^k} \\
=& \text{Pr}(b(D^{\prime})\in \Delta_3)\cdot \left(\frac{n+1}{n}\right)^k  \\
\leq& \text{Pr}(b(D^{\prime})\in \Delta)\cdot \left(\frac{n+1}{n}\right)^k +  \left(1 - \left(\frac{n-1}{n}\right)^k\right) .
\end{align}

Similarly, when $|D^{\prime}|=n-1$, we have the following:
\begin{align}
 & \text{Pr}(b(D)\in \Delta) \\
=& \text{Pr}(b(D)\in \Delta_1) + \text{Pr}(b(D)\in \Delta_2) + \text{Pr}(b(D)\in \Delta_3) \\   
=& \text{Pr}(b(D)\in \Delta_3) + \text{Pr}(b(D)\in \Delta_1) \\
\leq & \text{Pr}(b(D)\in \Delta_3) + \left(1 - \left(\frac{n-1}{n}\right)^k\right) \\
=& \text{Pr}(b(D^{\prime})\in \Delta_3)\cdot \frac{|D^{\prime}|^k}{|D|^k} + \left(1 - \left(\frac{n-1}{n}\right)^k\right) \\
=& \text{Pr}(b(D^{\prime})\in \Delta_3)\cdot \left(\frac{n-1}{n}\right)^k + \left(1 - \left(\frac{n-1}{n}\right)^k\right) \\
\leq & \text{Pr}(b(D^{\prime})\in \Delta)\cdot \left(\frac{n-1}{n}\right)^k + \left(1 - \left(\frac{n-1}{n}\right)^k\right) .
\end{align}

Since we have $(\frac{n-1}{n})^k \leq (\frac{n+1}{n})^k$, we have the following when $|D^{\prime}|=n\pm 1$: 
\begin{align}
\forall \Delta \in \Phi,    \text{Pr}(b(D)\in \Delta) \leq \text{Pr}(b(D^{\prime})\in \Delta)\cdot \left(\frac{n +1}{n}\right)^k + \left(1 - \left(\frac{n-1}{n}\right)^k\right). 
\end{align}

We let $e^\epsilon = (\frac{n +1}{n})^k$. Then, we have $\epsilon =k\cdot \ln{\frac{n + 1}{n}}$. Therefore, we show $b(D)$ achieves $(k\cdot \ln{\frac{n + 1}{n}}, 1 - (\frac{n-1}{n})^k)$-differential privacy. We can view the training of the base model as post-processing of $b(D)$, which does not alter the privacy level according to Proposition \ref{post-processing}. Therefore,  Bagging with replacement achieves $\left(k\cdot \ln{\frac{n + 1}{n}}, 1 - (\frac{n-1}{n})^k\right)$-differential privacy.



\section{Proof of Theorem~\ref{theorem4}}
\label{appendix-C}

We show Bagging without replacement achieves $\left(\ln{\frac{n+1}{n+1-k\cdot N}}, \frac{k\cdot N}{n} \right)$-differential privacy when training $N$ models with $N\cdot k$ examples. We first show the proof for $N=1$, and then generalize it to arbitrary $N$. We reuse some notations defined in the proof of Theorem~\ref{bagging_dp_theorem_2.2} for simplicity. We use operator $a$ to denote the sampling operation of Bagging without replacement. We use $\Psi$ to denote the joint domain of $a(D)$ and $a(D^{\prime})$. Specifically, each element $\psi \in \Psi$ is a subsample of $k$ examples. Given $a(D)$ and $a(D^{\prime})$, we denote the following three subsets of domain $\Psi$: 
\begin{align}
   &  \Lambda_{1} = \{\psi \in \Psi| \text{Pr}(a(D)=\psi)>0, \text{Pr}( a(D^{\prime})=\psi) = 0\}, \\
   &  \Lambda_{2} = \{\psi \in \Psi| \text{Pr}(a(D)=\psi)=0, \text{Pr}( a(D^{\prime})=\psi) > 0\}, \\
   &  \Lambda_{3} = \{\psi \in \Psi| \text{Pr}(a(D)=\psi)>0, \text{Pr}( a(D^{\prime})=\psi) > 0\}. 
\end{align}

Intuitively, $\Lambda_{1}$ includes all subsamples that can only be derived from $D$. $\Lambda_{2}$ includes all subsamples that can only be derived from $D^{\prime}$. $\Lambda_{3}$ includes all subsamples that can be derived from both $D$ and $D^{\prime}$. We note that $\Lambda_{1}\cup \Lambda_{2} \cup \Lambda_{3}$ constitutes all possible subsamples of $a(D)$ and $a(D^{\prime})$. Since we consider Bagging without replacement, we have the following: 
\begin{align}
 &   \text{Pr}(a(D)=\psi)= 
    \begin{cases}
     \frac{1}{{|D|\choose k}}, &\text{ if } \psi \in \Lambda_1 \cup \Lambda_3, \\
     0, &\text{ otherwise}. 
    \end{cases} \\
&    \text{Pr}(a(D^{\prime})=\psi)= 
    \begin{cases}
     \frac{1}{{|D^{\prime}|\choose k}}, &\text{ if } \psi \in \Lambda_2 \cup \Lambda_3, \\
     0, &\text{ otherwise}. 
    \end{cases}
\end{align}

Moreover, the number of elements in $\Lambda_1$, $\Lambda_2$, and $\Lambda_3$ can be computed as ${|D| \choose k}-{|D\cap D^{\prime}|\choose k}$, ${|D^{\prime}| \choose k}-{|D\cap D^{\prime}| \choose k}$, and ${|D\cap D^{\prime}| \choose k}$, respectively.  
Given the probability density and the size of $\Lambda_1$, $\Lambda_2$, and $\Lambda_3$, we have the following:
\begin{align}
 &   \text{Pr}(a(D)\in \Lambda_{3}) = \frac{{| D \cap D^{\prime}| \choose k}}{{|D| \choose k}},
    \text{Pr}(a(D)\in \Lambda_{1}) = 1 -  \frac{{| D \cap D^{\prime}| \choose k}}{{|D| \choose k}} ,
     \text{Pr}(a(D)\in \Lambda_{2}) = 0. \\
 &   \text{Pr}(a(D^{\prime})\in \Lambda_{3}) = \frac{{| D \cap D^{\prime}| \choose k}}{{|D^{\prime}| \choose k}},
    \text{Pr}(a(D^{\prime})\in \Lambda_{2}) = 1 -  \frac{{| D \cap D^{\prime}| \choose k}}{{|D^{\prime}| \choose k}},
     \text{Pr}(a(D^{\prime})\in \Lambda_{1}) = 0.
\end{align}

Recall our goal is to show Bagging without replacement achieves $(\epsilon,\delta)$-differential privacy. We will first find a $\delta$ and then determine the value of $\epsilon$ for the given $\delta$. In particular, for any $\Delta \subseteq \Lambda_1$ and $|D|=n-1$, we have the following: 
\begin{align}
    \text{Pr}(a(D)\in \Delta) \leq \text{Pr}(a(D)\in \Lambda_1) =1-\frac{{n-1 \choose k}}{{n \choose k}}=\frac{k}{n}, \text{Pr}(a(D^{\prime})\in \Delta) = 0. 
\end{align}

In other words, we have the following:
\begin{align}
 \forall \Delta \subseteq \Lambda_1,  \text{Pr}(a(D)\in \Delta) \leq e^{\epsilon}\cdot 0 + \frac{k}{n} =e^{\epsilon}\cdot  \text{Pr}(a(D^{\prime})\in \Delta)+ \frac{k}{n}.
\end{align}

Therefore, we let $\delta = \frac{k}{n}$. Next, we prove that for $\forall \Delta \subseteq \Psi$, i.e., $\Delta$ is an arbitrary subset of $\Psi$, we have the following:
\begin{align}
\label{appendix-37}
    \text{Pr}(a(D)\in \Delta) \leq \frac{n+1}{n+1-k} \cdot \text{Pr}(a(D^{\prime})\in \Delta) + \frac{k}{n}.
\end{align}

We decompose $\Delta$ into three disjoint subsets, i.e., $\Delta =\Delta_{1}\cup \Delta_2 \cup \Delta_3$, where $\Delta_1 \subseteq \Lambda_1$, $\Delta_2 \subseteq \Lambda_2$, and $\Delta_3 \subseteq \Lambda_3$. We first consider $|D^{\prime}|=n+1$. Then, we have the following: 
\begin{align}
& \text{Pr}(a(D)\in \Delta) \\
=& \text{Pr}(a(D)\in \Delta_1) + \text{Pr}(a(D)\in \Delta_2) + \text{Pr}(a(D)\in \Delta_3) \\
=& \text{Pr}(a(D)\in \Delta_3) \\
=& \text{Pr}(a(D^{\prime})\in \Delta_3)\cdot \frac{{|D^{\prime}| \choose k}}{{|D| \choose k}} \\
=& \text{Pr}(a(D^{\prime})\in \Delta_3)\cdot \frac{{n+1 \choose k}}{{n \choose k}}  \\
\leq& \text{Pr}(a(D^{\prime})\in \Delta_3)\cdot \frac{n+1}{n+1-k} \\
\leq& \text{Pr}(a(D^{\prime})\in \Delta)\cdot \frac{n+1}{n+1-k} +  \frac{k}{n} .
\end{align}

Similarly, when $|D^{\prime}|=n-1$, we have the following:
\begin{align}
 & \text{Pr}(a(D)\in \Delta) \\
=& \text{Pr}(a(D)\in \Delta_1) + \text{Pr}(a(D)\in \Delta_2) + \text{Pr}(a(D)\in \Delta_3) \\   
=& \text{Pr}(a(D)\in \Delta_3) + \text{Pr}(a(D)\in \Delta_1) \\
\leq & \text{Pr}(a(D)\in \Delta_3) + \frac{k}{n} \\
=& \text{Pr}(a(D^{\prime})\in \Delta_3)\cdot \frac{{|D^{\prime}| \choose k}}{{|D| \choose k}} + \frac{k}{n} \\
=& \text{Pr}(a(D^{\prime})\in \Delta_3)\cdot \frac{n-k}{n} + \frac{k}{n} \\
\leq & \text{Pr}(a(D^{\prime})\in \Delta)\cdot \frac{n-k}{n} + \frac{k}{n} .
\end{align}

Since we have $\frac{n-k}{n} \leq \frac{n+1}{n+1-k}$, we have the following when $|D^{\prime}|=n\pm 1$: 
\begin{align}
\forall \Delta \in \Psi,    \text{Pr}(a(D)\in \Delta) \leq \text{Pr}(a(D^{\prime})\in \Delta)\cdot \frac{n+1}{n+1-k} + \frac{k}{n} . 
\end{align}

We let $e^\epsilon = \frac{n+1}{n+1-k}$. Then, we have $\epsilon = \ln{\frac{n+1}{n+1-k}}$. Therefore, we show $a(D)$ achieves $( \ln{\frac{n+1}{n+1-k}}, \frac{k}{n} )$-differential privacy. We can view the training of the base model as post-processing of $a(D)$, which does not alter the privacy level. Therefore, Bagging without replacement achieves $(\ln{\frac{n+1}{n+1-k}}, \frac{k}{n} )$-differential privacy when $N=1$. Similarly, when we sample $N\cdot k$ examples to create a subsample, $a(D)$ achieves $(\ln{\frac{n+1}{n+1-N\cdot k}}, \frac{N\cdot k}{n} )$-differential privacy. Given $N\cdot k$ examples, we can evenly divide them into $N$ subsamples, and train $N$ base models which does not incur extra privacy loss. Therefore, Bagging without replacement achieves $(\ln{\frac{n+1}{n+1-N\cdot k}}, \frac{N\cdot k}{n} )$-differential privacy when training $N$ base models.

\section{Proof of Theorem~\ref{tightness-theorem-3}}
\label{appendix-D}
We show the tightness of our derived $(N\cdot k\cdot \ln{\frac{n + 1}{n}}, 1 - (\frac{n-1}{n})^{N\cdot k})$-differential privacy for Bagging with replacement when no assumptions are made on the base learner. Without loss of generality, we show the proof for $N=1$. Specifically, for $N>1$, we can use $N\cdot k$ as a new subsample and obtain the proof for arbitrary $N$. Our idea is that, for any $\delta< 1 - (\frac{n-1}{n})^{k}$, we can find a base learner such that Bagging with replacement cannot satisfy $(\epsilon,\delta)$-differential privacy for any $\epsilon$. For simplicity, we reuse the notations in the proof of Theorem~\ref{bagging_dp_theorem_2.2}. 
We let $|D^{\prime}|=n-1$ and  $\Delta^{\prime}=\Gamma_1$. Then we have the following probabilities: 
\begin{align}
  &  \text{Pr}(b(D)\in \Delta^{\prime})=\text{Pr}(b(D)\in \Gamma_1) = 1 -  \left(\frac{\mid D \cap D^{\prime}\mid}{\mid D\mid}\right)^{k} = 1 - \left(\frac{n-1}{n}\right)^k, \\
 &   \text{Pr}(b(D^{\prime})\in \Delta^{\prime})=\text{Pr}(b(D^{\prime})\in \Gamma_1) = 0. 
\end{align}

In other words, when  $\delta< 1 - (\frac{n-1}{n})^k$, we have the following: 
\begin{align}
 \text{Pr}(b(D)\in \Delta^{\prime}) > e^{\epsilon}\cdot 0 + \delta=e^{\epsilon}\cdot  \text{Pr}(b(D^{\prime})\in \Delta^{\prime}) + \delta, \text{ for any } \epsilon.    
\end{align}

Therefore, if the base learner is the $k$-nearest neighbor algorithm, whose model parameters are the $k$ training examples in the subsample, Bagging with replacement cannot achieve $(\epsilon,\delta)$-differential privacy.  Therefore, when no assumptions are made on the base learner, Bagging with replacement cannot achieve $(\epsilon ,\delta )$-differential privacy when $\delta <  1 - (\frac{n-1}{n})^k$.

\section{Proof of Theorem~\ref{tightness-theorem-4}}
\label{appendix-E}

We show the tightness of our derived $(\ln{\frac{n+1}{n+1-N\cdot k }}, \frac{N\cdot k}{n} )$-differential privacy for Bagging without replacement when no assumptions are made on the base learner. Without loss of generality, we show the proof for $N=1$. Specifically, for $N>1$, we can use $N\cdot k$ as a new subsample and obtain the proof for arbitrary $N$. Our idea is that, for any $\delta< \frac{N\cdot k}{n}$, we can find a base learner such that Bagging without replacement cannot satisfy $(\epsilon,\delta)$-differential privacy for any $\epsilon$. For simplicity, we reuse the notations in the proof of Theorem~\ref{theorem4}. We let $|D^{\prime}|=n-1$ and $\Delta^{\prime}=\Lambda_1$. 

Then, we have the following probabilities: 
\begin{align}
  &  \text{Pr}(a(D)\in \Delta^{\prime})=\text{Pr}(a(D)\in \Lambda_1) = 1 -  \frac{{| D \cap D^{\prime}| \choose k}}{{|D| \choose k}} = \frac{k}{n}, \\
 &   \text{Pr}(a(D^{\prime})\in \Delta^{\prime})=\text{Pr}(a(D^{\prime})\in \Lambda_1) = 0. 
\end{align}

In other words, when $\delta <  \frac{k}{n}$, we have the following: 
\begin{align}
 \text{Pr}(a(D)\in \Delta^{\prime}) > e^{\epsilon }\cdot 0 + \delta =e^{\epsilon }\cdot  \text{Pr}(b(D^{\prime})\in \Delta ) + \delta,\text{ for any }\epsilon.     
\end{align}

Therefore, if the base learner is the $k$-nearest neighbor algorithm, whose model parameters are the $k$ training examples in the subsample, Bagging without replacement cannot achieve $(\epsilon,\delta)$-differential privacy.  Therefore, when no assumptions are made on the base learner, Bagging without replacement cannot achieve $(\epsilon ,\delta )$-differential privacy when  $\delta <  \frac{k}{n}$.

\section{Model Architecture for MNIST}

\begin{table}[h]
  \caption{Model architecture of the base classifier for MNIST experiments.}
  \label{MNIST-structure}
  \centering
  
  \begin{tabular}{cc}
    \toprule
    Layer & Output Shape \\
    \midrule
    (Input) & $28 \times 28 \times 1$\\
    2D Convolution + ReLU & $14 \times 14 \times 16$\\
    2D MaxPooling & $13 \times 13 \times 16$\\
    2D Convolution + ReLU & $5 \times 5 \times 32$\\    
    2D MaxPooling & $4 \times 4 \times 32$\\    
    Flatten & $512$\\
    Fully Connected + ReLU & $32$\\
    Fully Connected & 10 \\
    \bottomrule
  \end{tabular}
\end{table}


\end{document}